\documentclass{wscpaperproc}
\usepackage{latexsym}
\usepackage{graphicx}
\usepackage{mathptmx}
\usepackage{enumitem}
\usepackage[mathscr]{euscript}
\usepackage[ruled,noline]{algorithm2e}
\usepackage{subcaption}
\usepackage{amsmath}
\usepackage{amsfonts}
\usepackage{amssymb}
\usepackage{amsbsy}
\usepackage{amsthm}
\usepackage{bm}
\usepackage{xspace}
\usepackage{framed}

\usepackage[pdftex,colorlinks=true,urlcolor=black,citecolor=black,anchorcolor=black,linkcolor=black]{hyperref}
\hypersetup{nolinks=false}

\newtheorem{assumption}{Assumption}

\newtheorem{example}{Example}
\newtheorem{lemma}{Lemma}

\hypersetup{colorlinks = true, urlcolor = blue, linkcolor = black,
  citecolor = black }
\newcommand{\xmath}[1]{\ensuremath{#1}\xspace}

\newcommand{\R}{\xmath{\mathbb{R}}}
\newcommand{\mv}[1]{\bm{#1}}
\newcommand{\ttheta}{\mv{\theta}}

\newcommand{\rr}{\mv{r}}

\newcommand{\xx}{\mv{x}}
\newcommand{\XX}{\mv{X}}

\newcommand{\ZZ}{\mv{Z}}
\newcommand{\yy}{\mv{y}}
\newcommand{\zz}{\mv{z}}
\newcommand{\setindicator}{\mv{1}}

\newcommand{\Real}{\mathbb{R}}

\newcommand{\fis}{\hat{f}_{{\tt is},n}}

\newtheoremstyle{wsc}
{3pt}
{3pt}
{}
{}
{\bf}
{}
{.5em}
{}

\theoremstyle{wsc}
\newtheorem{theorem}{Theorem}

\newtheorem{definition}{Definition}

\newtheorem{formulation}[theorem]{Formulation}

\begin{document}

%
%

\pagestyle{fancyplain}

\thispagestyle{plain}
\firstPageHead{}

\chead{\fancyplain{}{\itshape Deo and Murthy}}

\rhead{}
\cfoot{}
\renewcommand{\headrulewidth}{0pt} 

\makeatletter
\let\@internalcite\cite
\def\cite{\def\@citeseppen{-1000}%
    \def\@cite##1##2{(##1\if@tempswa , ##2\fi)}%
    \def\citeauthoryear##1##2##3{##1 ##3}\@internalcite}
\def\citeNP{\def\@citeseppen{-1000}%
    \def\@cite##1##2{##1\if@tempswa , ##2\fi}%
    \def\citeauthoryear##1##2##3{##1 ##3}\@internalcite}
\def\citeN{\def\@citeseppen{-1000}%
    \def\@cite##1##2{##1\if@tempswa, ##2)\else{}\fi}%
    \def\citeauthoryear##1##2##3{##1 (##3)}\@citedata}
\def\citeA{\def\@citeseppen{-1000}%
    \def\@cite##1##2{(##1\if@tempswa , ##2\fi)}%
    \def\citeauthoryear##1##2##3{##1}\@internalcite}
\def\citeANP{\def\@citeseppen{-1000}%
    \def\@cite##1##2{##1\if@tempswa , ##2\fi}%
    \def\citeauthoryear##1##2##3{##1}\@internalcite}
\def\shortcite{\def\@citeseppen{-1000}%
    \def\@cite##1##2{(##1\if@tempswa , ##2\fi)}%
    \def\citeauthoryear##1##2##3{##2 ##3}\@internalcite}
\def\shortciteNP{\def\@citeseppen{-1000}%
    \def\@cite##1##2{##1\if@tempswa , ##2\fi}%
    \def\citeauthoryear##1##2##3{##2 ##3}\@internalcite}
\def\shortciteN{\def\@citeseppen{-1000}%
    \def\@cite##1##2{##1\if@tempswa, ##2\else{}\fi}%
    \def\citeauthoryear##1##2##3{##2 (##3)}\@citedata}
\def\shortciteA{\def\@citeseppen{-1000}%
    \def\@cite##1##2{(##1\if@tempswa , ##2\fi)}%
    \def\citeauthoryear##1##2##3{##2}\@internalcite}
\def\shortciteANP{\def\@citeseppen{-1000}%
    \def\@cite##1##2{##1\if@tempswa , ##2\fi}%
    \def\citeauthoryear##1##2##3{##2}\@internalcite}
\def\citeyear{\def\@citeseppen{-1000}%
    \def\@cite##1##2{(##1\if@tempswa , ##2\fi)}%
    \def\citeauthoryear##1##2##3{##3}\@citedata}
\def\citeyearNP{\def\@citeseppen{-1000}%
    \def\@cite##1##2{##1\if@tempswa , ##2\fi}%
    \def\citeauthoryear##1##2##3{##3}\@citedata}
%
%
%
\def\@citedata{%
    \@ifnextchar [{\@tempswatrue\@citedatax}%
                  {\@tempswafalse\@citedatax[]}%
}

\def\@citedatax[#1]#2{%
\if@filesw\immediate\write\@auxout{\string\citation{#2}}\fi%
  \def\@citea{}\@cite{\@for\@citeb:=#2\do%
    {\@citea\def\@citea{, }\@ifundefined
       {b@\@citeb}{{\bf ?}%
       \@warning{Citation `\@citeb' on page \thepage \space undefined}}%
{\csname b@\@citeb\endcsname}}}{#1}}%

%
\def\@citex[#1]#2{%
\if@filesw\immediate\write\@auxout{\string\citation{#2}}\fi%
  \def\@citea{}\@cite{\@for\@citeb:=#2\do%
    {\@citea\def\@citea{; }\@ifundefined
       {b@\@citeb}{{\bf ?}%
       \@warning{Citation `\@citeb' on page \thepage \space undefined}}%
{\csname b@\@citeb\endcsname}}}{#1}}%

%
\def\@biblabel#1{}
\makeatother



\newdimen\bibindent
\bibindent=0.0em
\def\thebibliography#1{\section*{\refname}\list
   {}{\settowidth\labelwidth{[#1]}
   \leftmargin\parindent
   \itemindent -\parindent
   \listparindent \itemindent
   \itemsep 0pt
   \parsep 0pt}
   \def\newblock{}
   \sloppy
   \sfcode`\.=1000\relax}


\setlength{\baselineskip}{12.7pt}

\title{Importance Sampling for Minimization of Tail Risks: A Tutorial}

\author{Anand Deo\\[12pt]
	Indian Institute of Management Bangalore\\
	Bangerghatta Road, Bilekahalli\\
	Bangalore, 560076, INDIA\\
\and
Karthyek Murthy\\[12pt]
Singapore University of Technology and Design\\
8 Somapah Road\\
Singapore 487372, SINGAPORE
}

\maketitle
\section*{ABSTRACT}
This paper provides an introductory overview of how one may employ importance sampling effectively as a tool for solving stochastic optimization formulations incorporating tail risk measures such as Conditional Value-at-Risk. Approximating the tail risk measure by its sample average approximation, while appealing due to its simplicity and universality in use, requires a large number of samples to be able to arrive at risk-minimizing decisions with high confidence. This is primarily due to the rarity with which the relevant tail events get observed in the samples. In simulation, Importance Sampling  is among the most prominent methods for substantially reducing the sample requirement while estimating probabilities of  rare events. Can importance sampling be used for optimization as well? If so, what are the ingredients required for making importance sampling an effective tool for optimization formulations involving rare events? This tutorial aims to provide an introductory overview of the two key ingredients in this regard, namely, (i) how one may arrive at an importance sampling change of measure prescription at every decision, and (ii) the prominent techniques available for integrating such a prescription within a solution paradigm for stochastic optimization formulations.  

\section{INTRODUCTION}
When building optimization models for planning and  decision-making under uncertainty, a risk-neutral approach would seek to identify the best decision on average. Relying on independent yet similar repeated chances, a risk-neutral approach is justified by the law of large numbers. An unfortunate run entailing unacceptable losses in the first few runs, for example, may however make it harder to sustain and continue the operations even if an eventual upturn is inevitable. Therefore, the practice of risk management strives to look beyond expected outcomes and proactively shape the loss distribution, particularly aspects such as variability, risks posed by extreme losses, etc. 

Mean-risk optimization models, which seeks to minimize a measure of risk while meeting a target mean return, is a prominent approach by which a modeler may introduce risk-aversion while  optimizing under uncertainty. When performing optimization, a convex risk measure like conditional-value at risk (CVaR) becomes  particularly appealing due to its ability to quantitatively capture distribution tail  risks while retaining the convexity of the objective. Roughly speaking, CVaR at a quantile level $1-\beta$ captures the loss due to top $\beta$-fraction of the samples. Since the introduction of CVaR for optimization under uncertainty in \citeN{rockafellar2000optimization} and \citeN{uryasev}, mean-CVaR optimization modeling has become one of the most common vehicle for managing risk in numerous operations research applications, and as well in a number of related engineering disciplines. 

As with most stochastic optimization formulations, solving a mean-CVaR optimization model is typically tackled by approximating the mean and CVaR by their respective Monte Carlo sample average approximations (SAA).  
Despite the bottleneck that only a small fraction of the samples contribute to the evaluation of the CVaR criterion, SAA remains the most preferred solution approach due to its simplicity and near-universality in use: The SAA procedure and methods for inferring its solution quality remain unchanged as long as the objective and constraints possess finite second moments (see e.g., \shortciteNP{shapiro1991asymptotic,homem2014monte}). 
For a risk-averse optimization formulation involving CVaR at quantile level $1-\beta,$ the sample requirement for SAA to work gets blown up by a large multiplicative factor of $O(1/\beta),$ as $\beta \rightarrow 0,$ when compared to the risk-neutral counterparts. This unfortunately  leads to extraordinarily large formulations if one is employing a deterministic solver, and slower convergence if one is employing stochastic gradients. 

If we view the literature on tackling rare tail events in simulation, we witness Importance Sampling (IS) as one of the most prominent variance reduction approaches used for substantially reducing the sample requirement involved in estimating rare event probabilities and related expectations. Can importance sampling be used, to a similar degree of effectiveness, in optimization as well? This tutorial is dedicated to concisely introducing the ingredients required for using importance sampling for optimization under rare events. As the literature on simulation of rare events is rich and several beautiful expository reviews have been written on importance sampling for rare events, this tutorial will restrict its discussion to the scope of using IS for optimization, specifically for mean-CVaR formulations. The focus on CVaR is for the purposes of clarity, and the discussed methods extend even if CVaR is replaced by a different tail risk measure (such as expected excess loss) which preserves convexity.

The basic idea behind IS is to accelerate the occurrences of the tail risk events by sampling from alternate distributions which place greater emphasis on the risk scenarios of interest. Observed samples are then suitably reweighed to eliminate the bias introduced. This tutorial will specifically focus on the following two ingredients: (1) how one may arrive at an effective change of measure for evaluating at the objective at any given decision; and (2) how one may incorporate these decision-dependent changes of measure in a solution paradigm for minimizing CVaR.

The rest of the tutorial is organized as follows: Section 2 provides a definition of CVaR and introduces the risk-averse optimization formulation we shall be primarily considering in this paper. Section 3 provides an introduction to  importance sampling and discusses how one may arrive at an effective change of measure for a given decision with illustrative examples. Section 4 introduces a retrospective approximation approach, which can be understood as performing SAA with importance samples, while incorporating lazy-updates for changing the importance sampling distribution. Section 5 presents an adaptive stochastic approximation procedure which could be suitable if iterative gradient-descent methods are preferred. 
The methods are accompanied by results on the magnitude of sample reductions offered by efficient IS, relative to SAA, in  obtaining an optimal solution of desired quality.


\section{Risk-averse Optimization driven by Conditional Value at Risk}
In this section, we define the notion of conditional value at risk (CVaR), describe two prominent CVaR-driven optimization formulations, and briefly note  the merits and challenges in using SAA to solve them.  
\subsection{Conditional Value at Risk and its variational representation}
Let $\XX$ be a random vector modeling the collection of  uncertain variables affecting an optimization problem. Suppose that  $\ell(\xx,\ttheta)$ denotes the loss
incurred for a choice of decision $\ttheta$ when the random vector $\XX$ realises the value $\xx$. The Value at Risk (VaR) of the loss $\ell(\XX,\ttheta)$ at a tail level $\beta \in (0,1)$ is simply the loss quantile 
 \[v_\beta(\ttheta) = \inf\{u: P(\ell(\XX,\ttheta) \geq u) \leq \beta\}.\]
The Conditional Value at Risk (CVaR) of the loss $\ell(\XX,\ttheta)$ at a tail level $\beta$ is the average  loss given that $L(\XX,\theta)$ exceeds the respective value at risk: specifically,
\begin{equation}\label{eqn:CVaR_defintion}
    C_\beta(\ttheta) = E\left[\ell(\XX,\ttheta) \mid \ell(\XX,\ttheta) \geq v_\beta(\ttheta)\right].
\end{equation}
In this paper, we shall be considering the challenges in estimating and optimizing CVaR when the tail-level $\beta$ is close to zero. The following variational representation, due to \citeNP{rockafellar2000optimization}, makes CVaR conducive for optimization:
\begin{align}
    C_\beta(\ttheta) = \inf_{u\in \R} \left\{u+ \beta^{-1} E\left(\ell(\XX,\ttheta)-u\right)^+  \right\}.
    \label{eqn:cvar-var-rep}
\end{align}
Here $\left(\ell(\XX,\ttheta)-u\right)^+$ denotes the positive part $\max\{\ell(\XX,\ttheta)-u,0\},$ which captures the extent of excess loss above a level $u.$ The value at risk $v_\beta(\ttheta)$ is an optimal solution in the variational representation \eqref{eqn:cvar-var-rep} (see \cite{rockafellar2000optimization}), and it is readily verifiable that substituting the choice $u = v_\beta(\ttheta)$  in the objective in \eqref{eqn:cvar-var-rep} yields the right-hand side in \eqref{eqn:CVaR_defintion}.  

\subsection{Stochastic optimization formulations incorporating CVaR}
\begin{formulation}[Minimizing CVaR, potentially with constraints on the mean]
    Equipped with the above variational representation, if one wishes to minimize CVaR of a loss $\ell(\XX,\ttheta)$ over decision alternatives $\ttheta$ in the set $\Theta \subseteq \R^p,$  they may do so by solving the right-hand side of \eqref{eqn:CVaR_OpT} below.
\begin{align}\label{eqn:CVaR_OpT}
   c_\beta := \inf_{\ttheta \in \Theta} C_\beta(\ttheta) =\inf_{u\in\R,\,\ttheta\in \Theta} f(u,\ttheta),
\end{align}
where $f(u,\ttheta) = E[F(\XX \,;\, u,\ttheta)]$ and
\begin{align}
     F(\xx\,;\, u,\ttheta) := u+ \beta^{-1} \left(\ell(\XX,\ttheta)-u\right)^+  
     \label{eqn:F}
\end{align}
Observe that if the loss $\ell(\xx,\ttheta)$ is a convex function of $\ttheta,$ for any fixed $\xx,$ then the convexity is retained in \eqref{eqn:CVaR_OpT}. The set $\Theta$ can be modeled to include constraints on the decisions one may wish to impose, as illustrated in Example \ref{eg:lin-portfolio} below. 

\begin{example}[Portfolio optimization]
\label{eg:lin-portfolio} 
\textnormal{The task of constructing a linear portfolio with minimum risk while meeting a target return is among the simplest yet instructive examples one may consider. Suppose $\XX$ is an $\R^d$-valued random vector modeling the returns of $d$-assets. For a linear portfolio model  which places a weight $\theta_i$ over the asset $i,$ for $i = 1,\ldots,n,$ the return realization gets specified by $\ttheta^\intercal \XX.$ In this case, we can take the portfolio loss to be $\ell(\xx,\ttheta) = -\ttheta^\intercal \xx.$ It is convenient to require the weights placed over the $d$ assets to add up to 1. Therefore, when an investor seeks to meet a target return $t \in (0,+\infty),$ the constraint set $\Theta$ can be specified as in 
\begin{equation}
\label{eq:constraint}
    \Theta = \{ \ttheta \in \R^d_+: \boldsymbol{1}^\intercal \ttheta = 1, \ \boldsymbol{\mu}^\intercal \ttheta \geq t\},
\end{equation} 
where the vector $\boldsymbol{\mu}$ is the mean vector of the $d$ assets. }
\end{example}
\end{formulation}

\begin{formulation}[Mean-CVaR optimization]
Another convenient model for introducing risk aversion is to consider a convex combination of the CVaR criterion and the risk-neutral expected value objective as below: 
\begin{align}
    \inf_{\ttheta \in \Theta} \left\{ \lambda E[\ell(\XX,\ttheta)] + (1-\lambda)  C_\beta(\ttheta)\right\}, 
    \label{eqn:mean-cvar-opt}
\end{align}
where $\Theta$ is a convex subset of the euclidean space and $\lambda$ is a parameter governing the risk-appetite of a decision-maker. One may use a smaller value of $\lambda$ to specify a smaller appetite for risk.  From the variational representation in \eqref{eqn:cvar-var-rep}, the above mean-CVaR model simplifies to 
\begin{align*}
    \inf_{u \in R, \ttheta \in \Theta}  \left\{ \lambda E\left[ \ell(\XX,\ttheta) \right] +  (1-\lambda) \, E \left[F(\XX\,; \,  u,\ttheta)\right] \right\}, 
\end{align*}
where $F(\cdot)$ is defined, as before, in \eqref{eqn:F}. 
\end{formulation}

\begin{example}[Risk-averse two stage linear programs] 
\label{eg:two-stage-programs}
\textnormal{In two-stage  formulations, a  decision-maker takes an action in the first-stage; and in the wake of the realization of the random vector $\XX,$ he/she additionally gets to make a recourse decision augmenting the first-stage decision. Usually, a recourse decision is interpreted as utilizing the extra information to  compensate for any bad effects that might have been experienced as a result of first-stage action. A risk-averse two-stage linear program can be formulated as in 
 \eqref{eqn:mean-cvar-opt},  with 
\begin{align*}
    \ell(\xx,\ttheta) = c^\intercal \ttheta + Q(\xx,\ttheta),
\end{align*}
where $Q(\xx,\ttheta)$ is the optimal value of a second-stage linear program. The following is an example of a second-stage formulation:  
\begin{equation*}
    Q(\xx,\ttheta) =  \inf_{\substack{T\ttheta + W\yy= \mv{h}\\ \yy\geq \mv{0}}} \mv{y}^\intercal \xx
\end{equation*}
where $T$ and $W$ are suitably dimensioned matrices referred to as ``tender'' and ``recourse'' matrices. In this example, we have taken these matrices to deterministic. Allowing them to be random provides additional modeling power. One may refer to \shortciteNP[Chapters 1-2]{shapiro2021lectures}  for a comprehensive introduction to two-stage stochastic programming formulations and applications. Equipped with the flexibility to capture the effects of both the ``here-and-now'' decision $\ttheta$ and the recourse decision $\yy,$ two stage stochastic programming models are among the most widely  applied formulations for planning and decision-making under uncertainty.
}
\end{example}

\subsection{Sample-Average Approximation (SAA)}
For solving \eqref{eqn:CVaR_OpT}, one may consider its sample average approximation in which the expectations  $E(\ell(\XX,\ttheta) - u)^+$ in the objective, for all   $\ttheta \in \Theta,$ are replaced by the respective average over independent observations of $\XX.$ In particular, given $n$ i.i.d. samples of data $\XX_1,\ldots, \XX_n$ from the distribution of $\XX$, the sample averaged objective is denoted by,
\begin{equation}
\hat{f}_{n}(u,\ttheta) \ = \  u+\frac{1}{n\beta} \sum_{i=1}^n (\ell(\XX_i,\ttheta) - u)^{+}. 
\label{eq:saa-obj}
\end{equation}
Then a sample average approximation (SAA) to the optimization problem~\eqref{eqn:CVaR_OpT} may be specified as,
\begin{equation}\label{eqn:PMC}
\hat{c}_{n} = \inf_{u \in \R,\ttheta \in \Theta}  \hat{f}_n(u,\ttheta).
\end{equation}
Likewise, a sample-average approximation to the mean-CVaR formulation \eqref{eqn:mean-cvar-opt} is given by, 
\begin{align}
    \inf_{u \in \R, \ttheta \in \Theta} \left\{  \frac{\lambda}{n} \sum_{i=1}^n \ell(\XX_i,\ttheta) + (1-\lambda) \hat{f}_n(u,\ttheta) \right\}.
    \label{eqn:mean-cvar-saa}
\end{align}

Approximating the true expectations in the objective by their sample averages is the simplest approximation one can perform, and it has a broad appeal due to its simplicity and near-universal applicability. The (sample requirements for the) resulting formulations could be extraordinarily large however when handling tail expectations: Observe that the term $(\ell(\XX_i,\ttheta) - u)^{+}$ appearing in \eqref{eq:saa-obj} is often zero for most observations $\XX_i$, when a search is conducted over the variable $u$ to find its optimal value $v_\beta(\ttheta).$ This is due to the value-at-risk $v_\beta(\ttheta)$ being the $(1-\beta)$-th quantile of the loss distribution. Thus, a large fraction of terms in the summation in \eqref{eq:saa-obj} will be zero, which is in line with CVaR being a tail risk measure.  

Large sample properties such as consistency and  asymptotic normality are well-known for SAA estimators (see \citeNP[Theorem 3.2]{shapiro1991asymptotic}). An application of these properties to identify the number of samples $n$ required in \eqref{eqn:PMC} to approximate \eqref{eqn:CVaR_OpT} reveals that the number of samples required scales inversely proportional to the tail level $\beta$ of interest. This observation is consistent with the understanding that one would need approximately $\tilde{O}(\beta^{-1})$ samples, as $\beta \searrow 0,$ in order to witness loss scenarios exceeding the $(1-\beta)$-th quantile captured by the value at risk. 

To gather a sense of the magnitude of the number of samples required, consider the following portfolio optimization example from  \citeN{caccioli2018portfolio}: Even at a tail level of $\beta = 1/40,$ it has been observed that one would need about 14 years of observations to  achieve a 10\% relative error in the optimum portfolio's CVaR for 100 stocks. A more detailed discussion on this example is available in \citeN{caccioli2018portfolio}. The perils of minimizing CVaR with insufficient samples are also discussed with the help of a detailed empirical study in  \shortciteNP{lim2011conditional}.  The computational effort required for solving \eqref{eqn:PMC} becomes exorbitantly large as a consequence of the large sample requirement, specifically if the tail level $\beta$ is  small. Small values of $\beta$ may be particularly pertinent if a situation demands high reliability, as may be required in settings such as electric power dispatch and the design of cyber-physical systems.   

To keep the discussion centred on overcoming the challenges due to this rarity in CVaR minimization, we focus on the formulation \eqref{eqn:CVaR_OpT} in the rest of this tutorial and assume that the constraints in the decision set $\Theta$ are  specified in terms of explicitly known quantities: For example, if the constraint set is as in \eqref{eq:constraint}, we assume that the mean vector $\boldsymbol{\mu}$ is known, even though in practice, the mean has to be estimated by the empirical mean or its variants. The importance sampling techniques we develop in the subsequent sections are for approximating $E[F(\XX\,; \,u, \ttheta)],$ and they can be used for both the CVaR minimization \eqref{eqn:CVaR_OpT}, mean-CVaR \eqref{eqn:mean-cvar-opt} formulations, and as well their variants. In particular, we shall develop importance sampling based estimator $\fis(u,\ttheta)$ for  $E[F(\XX\,; \,u, \ttheta)],$  which can be used as a replacement for the sample average approximation $\hat{f}_n(u,\ttheta)$ in the same way in both the SAA objectives \eqref{eq:saa-obj} and \eqref{eqn:mean-cvar-saa}.

\section{Variance reduction with importance sampling}
\label{sec:IS}
This section introduces the basic idea behind importance sampling (IS) and outlines techniques one may employ to derive IS distributions for any fixed decision choice $\ttheta \in \Theta$.  
\subsection{Importance Sampling for tail estimation of tail risks}\label{sec:IS_Background}
Consider any fixed choice of $\ttheta \in \Theta.$ Recall that the main difficulty in the estimation of $E[(\ell(\XX,\ttheta)-u)^+]$ in \eqref{eqn:CVaR_OpT} is the lack of samples in the excess loss region $\{\ell(\xx,\ttheta)\geq u\},$ for values of $u$ around the $(1-\beta)$-th quantile of the loss. IS attempts to overcome this drawback by instead drawing samples from an alternative distribution under which this tail event occurs more frequently. The  bias incurred by sampling from a different distribution is compensated by suitably weighing the resulting observations.  
Specifically, let $\ZZ$ be another random vector of our choice whose probability density $f_{\ZZ}(\zz) >0$ is  absolutely continuous with respect to that of $\XX$: that is $f_{\ZZ}(\zz) >0$ whenever $f_{\XX}(\zz) >0$. Now, consider the following weighted estimator for the objective, 
\begin{equation}\label{eqn:pis}
     \fis(u,\ttheta) = u+ \frac{1}{n\beta} \sum_{i=1}^n  (\ell(\ZZ_i,\ttheta)-u)^+\frac{f_{\XX}(\ZZ_i)}{f_{\ZZ}(\ZZ_i)}
\end{equation}
where $\ZZ_{1},\ldots \ZZ_{n}$ are sampled i.i.d. from the distribution of $\ZZ$. The ``weights" in the above estimators are likelihood ratios ${f_{\XX}(\ZZ_i)}/{f_{\ZZ}(\ZZ_i)}$ accompanying each observation of the excess loss $(\ell(\ZZ_i,\ttheta) - u)^+.$ Had there been no change in the density from which the samples are obtained (that is, if the samples are obtained from the original density $f_{\XX}(\cdot)$ itself), then observe that we get back the SAA objective in \eqref{eq:saa-obj}. 

Suppose that $\fis$ has finite variance, and let any alternative  distribution choice  $f_{\ZZ}$ which possess this property be labeled as ``admissible". Then observe that as the number of samples increase, the IS objective $\fis(u,\ttheta)$ approximates the desired objective $f(u,\ttheta) = u + \beta^{-1}E[(\ell(\XX,\ttheta)-u)^+]$ due to the following: 
\begin{align*}
    \fis(u,\ttheta) &\to u +  \beta^{-1} E\left[(\ell(\ZZ,\ttheta)-u)^+\frac{f_{\XX}(\ZZ)}{f_{\ZZ}(\ZZ)}\right]  = u + \beta^{-1}\int_{\zz} (\ell(\zz,\ttheta)-u)^+ \frac{f_{\XX}(\zz)}{f_{\ZZ}(\zz)} f_{\ZZ}(\zz)d\zz \\
    &= u + \beta^{-1} \int_{\zz} (\ell(\zz,\ttheta)-u)^+ f_{\XX}(\zz) d\zz
    = u + \beta^{-1}E[(\ell(\XX,\ttheta)-u)^+] = f(u,\ttheta). 
\end{align*}
as $n\to\infty$, and therefore the estimator in \eqref{eqn:pis} is consistent.  In the above chain, the first equality holds due to law of large numbers. The above equations also show that there is no bias introduced in this approximation: That is, $E[\fis(u,\ttheta)] = f(u,\ttheta),$ for every choice of $u$ and $\ttheta.$ 

While every admissible  change of distribution approximates the target objective as above, the key to approximating well with a substantially smaller number of samples relies on  making a good choice for the IS distribution $\ZZ$. Indeed in the estimation of rare event probabilities, there is a considerable litreature on  how to arrive at good choices for IS distributions: See, for example, \shortciteNP{Heidelberger,AGbook,juneja2007asymptotics}, or, more recently \shortciteNP{blanchet2019rare,bai2022rare,deo2021efficient} for treatments on objectives which may have a greater relevance from an optimization point of view. 

Keeping typical objectives arising in  optimization in view, we next describe two  approaches which could be considered for arriving at an effective  IS distribution for any fixed decision $\ttheta \in \Theta.$ 

\subsection{Approach 1: IS via exponential twisting based on dominating points}\label{sec:IS_Twist}
The first approach for  deriving a good importance sampling change of measure is to explicitly use the distribution of $\XX$ and the loss $\ell(\XX,\ttheta)$ to carefully arrive at an IS distribution choice. 
Executing this approach typically involves two steps. 
\begin{itemize}[leftmargin = *]
    \item[] \textbf{Step 1:}  Use the log-moment generating function of $\XX$ to identify the so-called ``dominating points" of the excess loss set  $\{\xx: \ell(\xx,\ttheta) \geq u \},$ with reference to the given distribution for $\XX.$ Roughly speaking, the dominating points are a collection of  points in the excess loss set $\{ \xx: \ell(\xx,\ttheta) \geq u \}$ such that each dominating point captures, in a local sub-region of $\{ \xx: \ell(\xx,\ttheta) \geq u \}$, the most likely way the excess loss event happens. 
\item[] \textbf{Step 2:} Once the dominating points are identified, one typically chooses the IS distribution to be a mixture distribution, with each component distribution of the mixture being chosen to be an ``exponentially tilted'' distribution whose mean coincides with one of the dominating points. Thus one may need as many component distributions as the number of dominating points. 
\end{itemize}
By placing the emphasis on the dominating points, the guiding principle here is to ensure that locally within sub-regions in $\{\xx: \ell(\xx,\ttheta) \geq u \},$  the points which are more likely to be observed are indeed given more probability mass. 
See, for e.g.., \shortciteNP[Definition 2]{arief2021certifiable} for a  definition of dominating points. 
\begin{definition}[Exponentially tilted densities]
    Given a probability density $f_{\XX}$ for the random vector $\XX,$ we call a new density $g$ to be exponentially titled version of $f_{\XX}$ with  a tilt factor $\mv{b}$ if
    \begin{align*}
        g(\xx) \propto \exp\big( \mv{b}^\intercal \xx\big) f_{\XX}(\xx),
    \end{align*}
    and $E[\exp(\mv{b}^\intercal \XX)]$ is finite. In particular, we have 
    \begin{align}
        g(\xx) = \frac{\exp\left(  \mv{b}^\intercal \xx \right)}{E[\exp(\mv{b}^\intercal \XX)]} f_{\XX}(\xx) = \exp\left(  \mv{b}^\intercal \xx - \Lambda(\mv{b})\right)f_{\XX}(\xx),
        \label{eqn:IS-tilt-density}
    \end{align}
    where $\Lambda(\mv{r})$ is the log-moment generating function defined by $\Lambda(\rr) = \log E[\exp(\rr^\intercal \XX)],$ for $\rr \in \R^d.$
\end{definition}
\begin{lemma}
Suppose that the log-moment generating function $\Lambda(\rr) = E[\exp(\rr^\intercal \XX)]$ is finite and differentiable at $\rr = \mv{b}.$ If $\ZZ$ is distributed according to the exponentially tilted density $g(\cdot)$ in \eqref{eqn:IS-tilt-density}, then $E[\ZZ] = \nabla \Lambda(\mv{b}).$
\label{lem:exp-tilt-mean}
\end{lemma}
\begin{proof}\let\qed\relax
The conclusion follows directly from the following two deductions:
\begin{align*}
    E[\ZZ] &= \int \zz g(\zz) d\zz = \int \zz \frac{\exp(\mv{b}^\intercal \zz)}{E[\exp(\mv{b}^\intercal \XX)]}f_{\XX}(\zz)d\zz  = \frac{E\left[ \XX \exp(\mv{b}^\intercal \XX)\right]}{E\left[ \exp(\mv{b}^\intercal \XX)\right]}; \text{ and }\\
   \nabla \Lambda (\rr) &= \int \nabla \exp(\rr^\intercal \xx) g(\xx) d\xx = \int \xx  \exp(\rr^\intercal \xx) g(\xx) d\xx = \int \xx   \frac{\exp(\mv{r}^\intercal \xx)}{E[\exp(\mv{r}^\intercal \XX)]}f_{\XX}(\xx) d\xx = \frac{E\left[ \XX \exp(\mv{r}^\intercal \XX)\right]}{E\left[ \exp(\mv{r}^\intercal \XX)\right]}.
\end{align*}
\end{proof}
Equipped with the above definition of exponentially tilted density and its properties, Example~\ref{eg:Dominating_Points} below demonstrates how one may execute the two-step procedure indicated above.\\

\begin{example}\label{eg:Dominating_Points}\em
Consider the piece-wise linear loss, $\ell(\xx,\ttheta) = \max\{\ttheta^\intercal \mv{A}_i\xx:i = 1,\ldots,M\},$ where $\{\mv{A}_i: i = 1,\ldots,M\}$ are $p \times d$ matrices.  Let $\XX$ be a light-tailed with  log-moment generating function,   
$\Lambda(\rr) = \log E[\exp(\rr^\intercal \XX)],$
for $\rr \in \R^d.$ For this example, we take $\Lambda(\cdot)$ to be strictly convex, differentiable, and finite for every $\rr \in \R^d.$ An important quantity useful for identifying a good choice of IS distribution is the convex conjugate of the log-moment generating function $\Lambda(\cdot)$ defined as follows:
    $\Lambda^\ast (\xx) = \sup_{\rr \in \R^d}\left\{\rr^\intercal \xx - \Lambda(\rr)\right\}.$
In the example of $\XX$ being a multivariate normal random vector with mean $\mv{m}$ and covariance $\mv{\Sigma},$ we have $\Lambda(\rr) = \mv{m}^\intercal \mv{r} + \rr^\intercal \mv{\Sigma} \rr/2$ and $\Lambda^\ast(\xx) = (\xx - \mv{m})^\intercal \mv{\Sigma}^{-1}  (\xx - \mv{m})/2.$ 
For $i = 1,\ldots,M,$ let 
\begin{align}
    \mv{a}_i = \arg\min_{\xx} \left\{ \Lambda^\ast (\xx): \ttheta^\intercal \mv{A}_i\xx \geq u\right\}. 
    \label{eqn:FL_Log_MGF_arg-mins}
\end{align}
See that the excess loss set $\{\xx: \ell(\xx,\theta) \geq u \}$ which is of interest to us can be seen as the union of the sub-regions $\mathcal{R}_i := \{\xx: \ttheta^\intercal \mv{A}_i\xx \geq u\}.$ We next use $\{\mv{a}_i: i = 1,\ldots,M\}$ to find the corresponding collection of roots $\{\mv{b}_i: i = 1,\ldots,M\}$ satisfying 
\begin{equation}\label{eqn:IS_Twist}
    \nabla \Lambda(\mv{b}_i) = \mv{a}_i,
\end{equation} 
for $i=1,\ldots,M.$ With these definitions, the following two observations are in order: 
\begin{itemize}[leftmargin = *]
    \item[(i)] Due to the property of convex conjugates, we have that the convex conjugate of $\Lambda^\ast(\cdot)$ is, in turn, $\Lambda(\cdot)$ itself in this example. As a consequence, we also  symmetrically have $\nabla \Lambda^\ast(\mv{a}_i) = \mv{b}_i,$ for $i = 1,\ldots,M.$ Therefore, from the optimality conditions for $\mv{a}_i$ in \eqref{eqn:FL_Log_MGF_arg-mins}, we have for $i = 1,\ldots,M,$
    \begin{align*}
        \mv{b}_i^\intercal \big( \mv{\xx} - \mv{a}_i\big) \geq 0, \qquad \text{ for all } \xx \text{ in the sub-region } \mathcal{R}_i.
    \end{align*}
    \item[(ii)] From \eqref{eqn:IS_Twist} and Lemma \ref{lem:exp-tilt-mean}, we have $E[\ZZ] = \nabla \Lambda(\mv{b}_i) = \mv{a}_i$ when $\ZZ$ is distributed with an exponential tilting by a factor $\mv{b}_i$ as in, 
    \begin{align}
           g_i(\xx) \propto e^{\mv{b}_i^\intercal \xx}f_{\XX}(\xx).
           \label{eqn:mixture-comp}
    \end{align}
\end{itemize}
For the above reasons, the collections $\{\mv{a}_i: i = 1,\ldots,M\}$ and $\{\mv{b}_i: i = 1,\ldots,M\}$ are called as dominating points and tilt parameters, respectively. 
With the tilt factors computed as roots $\{\mv{b}_i: i =1,\ldots,M\}$ in \eqref{eqn:IS_Twist}, we select our IS distribution to be the mixture density 
\begin{align}
f_{\ZZ}(\xx) = \sum_{i=1}^M p_i g_i(\xx),    
\label{eqn:exp_twise}
\end{align}
in which the mixture component densities $g_i$ are obtained by exponentially tilting the original density $f_{\XX}(\cdot)$ by a factor $\mv{b}_i$ (as in \eqref{eqn:mixture-comp}) and the positive mixture weights satisfy $\sum_{i = 1}^m p_i = 1.$ The $i$-th component in the mixture density, $g_i(\cdot),$ has the dominating point $\mv{a}_i$ in the rare set as its mean, thereby placing a prominent amount of probability mass in the target rare set. Besides this property, the exponentially tilted densities $g_i$ are such that the respective likelihood ratios $f(\xx)/g_i(\xx)$ stay controlled throughout the rare sub-regions $\mathcal{R}_i = \{\xx:\ttheta^\intercal \mv{A}_i\xx_i \geq u\}.$  \hfill$\Box$
\end{example}

\noindent \textbf{Merits and challenges in executing IS via dominating points.} We shall see later in Section \ref{sec:red-IS-samples} that IS via dominating points, as in Example \ref{eg:Dominating_Points} above, reduces the sample requirements to a great degree, if (i) the distribution of $\XX$ is light-tailed, (ii) one can get hold of the dominating points in \eqref{eqn:FL_Log_MGF_arg-mins} and the tilting parameters  \eqref{eqn:IS_Twist} with relatively less effort, and (iii) there are not too many dominating points for any given choice of $\ttheta \in \Theta$. These conditions are  usually not met, though, for the following reasons:  
\begin{itemize}
\item[a)] One may often need to solve complicated, potentially nonconvex optimization problems described in terms of the log-moment generating function $\Lambda(\cdot),$ its convex conjugate $\Lambda^\ast$, and the excess loss set $\{\xx: \ell(\xx,\ttheta) \geq u \}$ to identify the dominating points. This is made further intractable by the requirements that (i) such complicated optimization problems (for identifying dominating points and tilt parameters) have to be solved repeatedly for different choices of $(u,\ttheta);$ and (ii) a description of log-moment generating function $\Lambda(\cdot)$ and its convex conjugate $\Lambda^\ast(\cdot)$ may not be available in most problems and have to be estimated additionally via Monte Carlo.  
\item[b)] Further, the number of dominating points may be quite large if the objective $L(\cdot)$ is complex (or) if the ambient dimension of $\XX$ is not small.   In such cases where the number of dominating points are large, identifying all of them and determining how to arrive at the weighting probabilities $p_i$ in \eqref{eqn:exp_twise} may be non-trivial. 
\end{itemize}
 While the above issues have somewhat limited the applicability of IS to a relatively narrow collection of instances, the last couple of years have witnessed efforts explicitly directed towards overcoming these limitations (see, example, \shortciteNP{he2023adaptive,deo2021achieving,arief2021certifiable} and references therein). Assuming access to the log-moment generating function $\Lambda(\cdot),$  attempts towards making IS based on dominating points applicable for more complex objectives have been undertaken  in \shortciteNP{arief2021certifiable} and \shortciteNP{probefficiency}. In particular, if the number of samples to be drawn is delicately chosen to be neither too small nor  too large, \shortciteNP{probefficiency} observes that it may be sufficient to choose the mixture distribution based on dominating points from the sub-collection $\text{argmin}\{\Lambda^\ast(\mv{a}_i): i = 1,\ldots,M\}.$ The approach presented in Section~\ref{sec:SS-IS} below, based on \shortciteNP{deo2021achieving}, is a radically different approach aiming  to overcome the difficulties (a) - (b) above by implicitly learning a good IS distribution from the samples of $\XX.$ 

\subsection{Approach 2: IS via Self-Structuring Transformations}
\label{sec:SS-IS}
In this approach, 
the search for effective IS distributions is instead recast as follows: ``Can we find a single transformation $\mv{T}(\cdot)$ whose respective push-forward distribution (i.e., the law of  $\mv{T}(\XX)$) readily serves as an effective IS distribution when deployed across a large class of problems?". This re-framed pursuit, seeking to induce an effective IS distribution \textit{implicitly} via a map $\mv{T}(\cdot),$  bypasses the need to explicitly tailor the IS distribution to every decision choice and to every problem. We shall see that this problem agnostic nature of the approach allows it to be simpler to use, making it closer in spirit to the SAA. The approach is sufficiently simple  to render itself to be readily applicable even for the two-stage programs in Example \ref{eg:two-stage-programs}. 

To explain the use of this IS approach towards solving \eqref{eqn:CVaR_OpT}, suppose that $\rho$ is a positive constant capturing the asymptotic growth rate of the loss $\ell(\xx,\ttheta)$ as a function of $\xx$: that is, $\lim_{n \rightarrow \infty} \ell(n\xx,\ttheta)/n^\rho > 0,$ for some $\ttheta \in \Theta$ and $\lim_{n \rightarrow \infty} \ell(n\xx,\ttheta)/n^\rho < +\infty,$ for all $\ttheta \in \Theta.$ For the piece-wise linear and two-stage losses in Examples \ref{eg:lin-portfolio} - \ref{eg:Dominating_Points}, we have $\rho = 1.$ For quadratic losses such as in the Delta-Gamma approximation for portfolio returns in \shortciteNP{glasserman2000variance}, we have $\rho = 2.$ Define the $\Real^d$-valued function $$\mv{T}_h(\xx) := \xx [s_{h}]^{\mv{\kappa}(\xx)},$$ where for  $h > 0,$ we take $s_{h} = h \max\{\log \log (1/\beta),1\}.$ The positive number $s_h$ can be viewed as a scalar stretch factor which allows the transformation $\mv{T}_h$ to stretch the different components of $\xx = (x_1,\ldots,x_d)$ differently via the vector-valued exponent $\mv{\kappa}(\xx) = (\kappa_1(\xx), \ldots,\kappa_d(\xx))$ defined as below: 
\begin{equation}\label{eqn:kappa_def}
{\kappa}_i(\xx) := \frac{\log (1+|x_i|)}{\rho \log(1+\Vert\xx\Vert_\infty)}, 
\quad i = 1,\ldots,d.
\end{equation}
The scalar stretch factor $s_h,$ when viewed as a function of tail level $\beta,$ is larger when the estimation problem is made rarer by letting $\beta$ smaller. 
Exponentiation is done component-wise in the above expression for $\mv{T}_h(\xx)$ as in, 
$\mv{T}_h(\xx) = (x_1 s_{h}^{\kappa_1(\xx)}, \ldots,x_d s_{h}^{\kappa_d(\xx)}).$
The map $\mv{T}_h:\Real^d \rightarrow \Real^d$ can be shown to be invertible almost everywhere on $\Real^d$ (see \citeNP[Proposition 1]{deo2021achieving}) and the transformed vector $\ZZ = \mv{T}_h(\XX)$ has a probability density if $\XX$ has a density. 
Letting $f_{\XX}$ and $f_{\ZZ}$ denote the respective densities of $\XX$  and $\ZZ,$ the likelihood ratio resulting from this change-of measure is given by, 
\begin{align}
\mathcal{L}_h = \frac{f_{\XX}(\ZZ)}{f_{\ZZ}(\ZZ)} = \frac{f_{\XX}(\ZZ)}{f_{\XX}(\XX)}  J_h(\XX)  
\label{eq:LR}
\end{align}
where the Jacobian, $J_h(\cdot),$ of the transformation equals,
\[
   J_h(\xx) =\left[\prod_{i=1}^d \tilde{J}_i(\xx) \right]\times \frac{s_h^{\mv{1}^\intercal \mv{\kappa}(\xx)}}{\max_{i=1,\ldots,d} \tilde{J}_i(\xx)}, \text{ with } \tilde{J}_i(\xx)
            := 1+\frac{\rho^{-1}\log(s_h)}{\log(1+\Vert \xx \Vert_\infty)
              }\frac{|x_i|}{1+|x_i|}, \quad i = 1,\ldots,d.
\]
From the i.i.d. samples $\ZZ_i = \mv{T}_h(\XX_i), i =1,\ldots,d,$ we have the following IS estimator for the objective function in \eqref{eqn:CVaR_OpT}:
\begin{equation}\label{eqn:FIS}
\fis(u,\ttheta) = \left[u+ \frac{1}{n\beta}\sum_{i=1}^n (\ell(\ZZ_{i},\ttheta) -u)^+\mathcal{L}_{h,i}\right],  
\end{equation}
where ${\XX}_1,\ldots,{\XX}_n$ are  i.i.d. copies of  $\XX$ and $\mathcal{L}_{h,i} = J_h(\XX_i)f_{\XX}(\ZZ_i)/f_{\XX}(\XX_i)$ denotes the likelihood  in  \eqref{eq:LR}. 


\begin{figure}[h!]
\centering
\begin{subfigure}{0.45\textwidth}
         \centering
\includegraphics[width=0.72\textwidth]{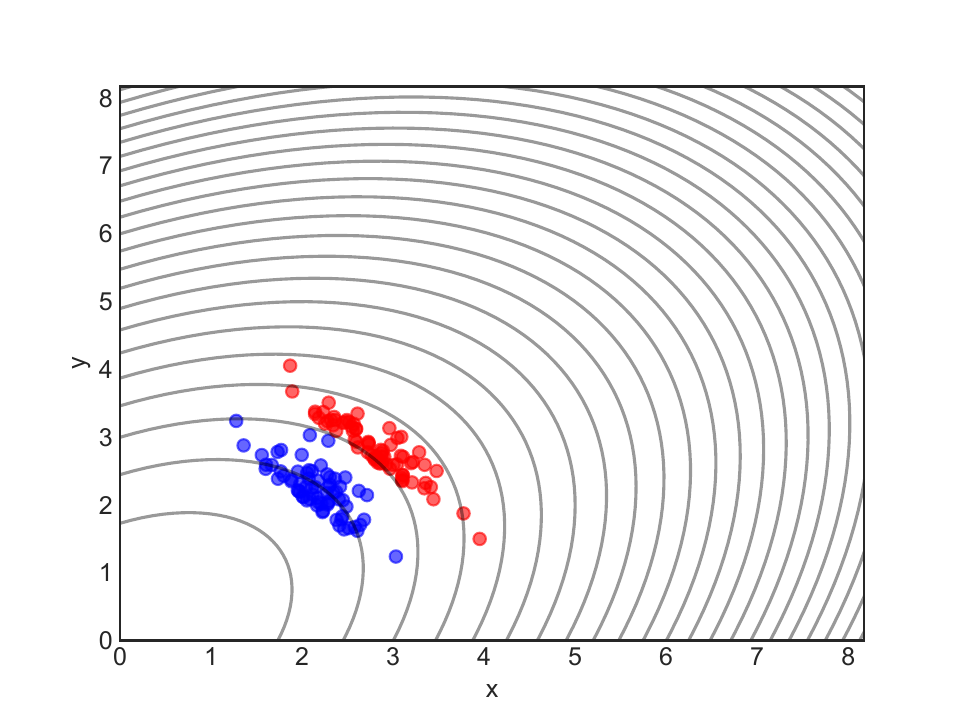} 
\end{subfigure}
\begin{subfigure}{0.45\textwidth}
         \centering
\includegraphics[width=0.72\textwidth]{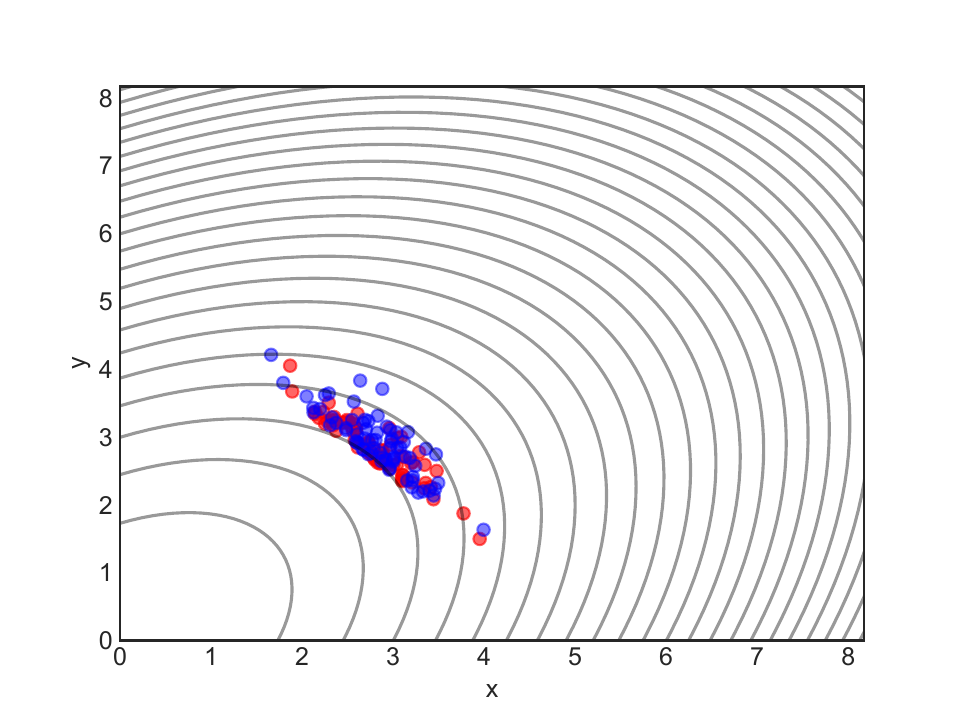} 
    \end{subfigure}
\begin{subfigure}{0.45\textwidth}
         \centering
\includegraphics[width=0.72\textwidth]{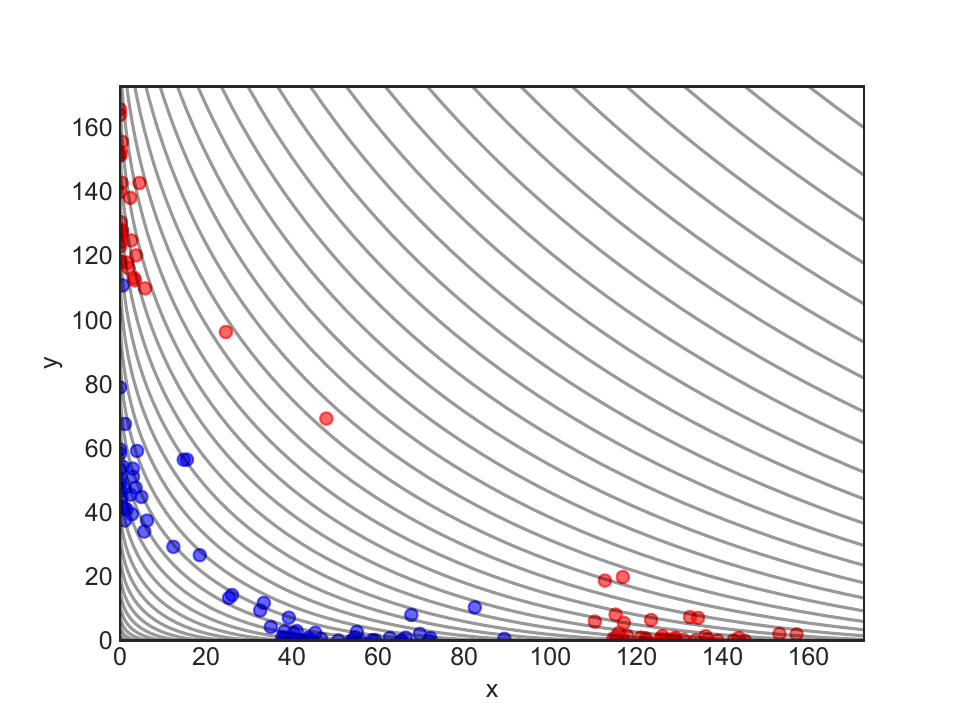} 
\end{subfigure}
\begin{subfigure}{0.45\textwidth}
         \centering
\includegraphics[width=0.72\textwidth]{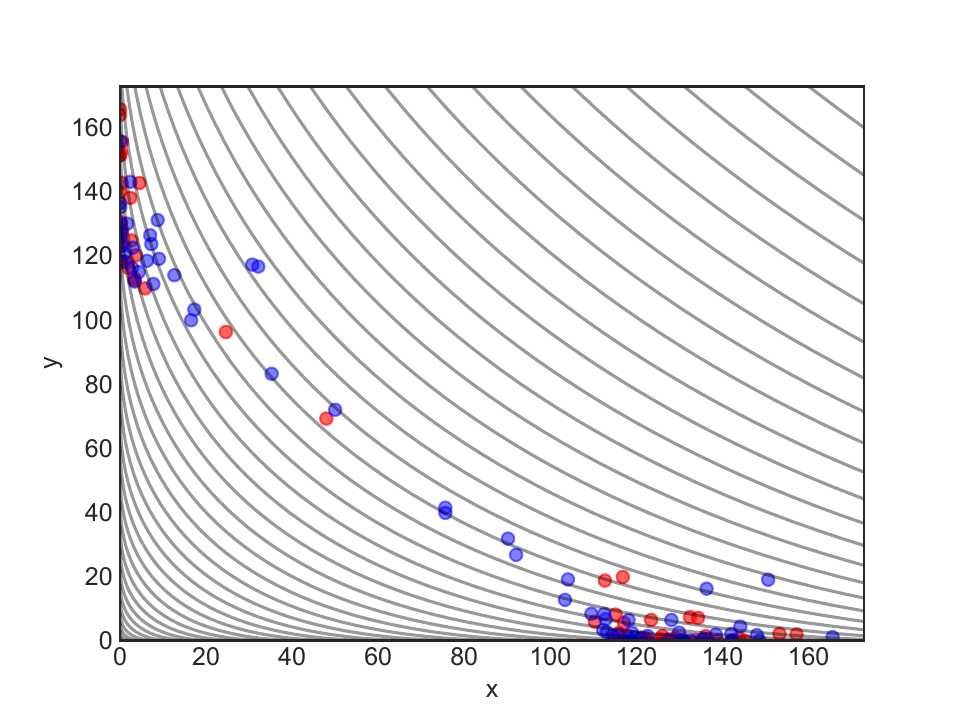} 
\end{subfigure}
\caption{An illustration of how the self-structuring transformation $\mv{T}_h(\cdot)$ helps induce IS distributions with desirable concentration properties:  The two panels in the left show the conditional excess loss samples of $\XX$ falling in the target tail set $\{\xx: \ell(\xx,\ttheta) \geq u\}$ (in red) and in the significantly less rare tail set $\{\xx: \ell(\xx,\ttheta) \geq l\},$ for a  level $l \ll u$ (in blue) and for two different choices of distributions for $\XX.$ The respective panels in the right show how the same transformation $\mv{T}_h,$ when applied to the blue samples in the top and bottom left panels, manages to replicate the concentrations of the respective target excess loss samples in red. 
}
\label{fig:SS-Transform}
\end{figure}

A main reason to use the map $\mv{T}_h(\cdot)$ is that it is capable of making use of the similarities in the distributions of the excess loss samples at different  levels of rarity in order to induce a good IS distribution. We devote the rest of this section towards developing an intuitive understanding of this property with the aid of Figure \ref{fig:SS-Transform} above. 
For values of $l$ smaller than $u,$  this approach builds on the premise that the excess loss samples falling in the target rare set $\{\xx:\ell(\xx,\ttheta)\geq u\}$ and those in the less rare set $\{\xx:\ell(\xx,\ttheta)\geq u\}$ exhibit a remarkable similarity in how they concentrate. The figures in the left panel in Figure \ref{fig:SS-Transform} illustrate this property by comparing the conditional samples of $\XX$ in the target tail event $\{\ell(\XX,\ttheta) \geq u\}$ (red points) against those in the less rare event $\{\ell(\XX,\ttheta) \geq l\}$ (blue points).  In the upper panel in the left , $\XX=(X_1,X_2)$ is drawn from a bivariate Gaussian distribution, while in the lower panel $\XX$ has i.i.d. Weibull marginals satisfying  $P(X_1 \leq x) =1-\exp(-x^{0.6})$.  For this illustration, the levels $u$ and $l$ are  chosen, respectively, to equal the $(1-10^{-4})$-th and $(1-10^{-2})$-th quantile of the loss; roughly speaking,  the blue samples falling in the the latter region get observed 100 times more frequently than the target excess loss samples in red. 
The similarity observed at different tail levels in Figure \ref{fig:SS-Transform} can be shown to hold quite generally, including various parametric and semiparametric distributions and copula families commonly used in practice (see Assumption~\ref{assume:BBIS} in Section \ref{sec:red-IS-samples} and \shortciteNP[Table 2]{deo2021achieving} for a comprehensive collection of  distribution families satisfying this tail-similarity property).

The map $\mv{T}_h(\cdot)$ is such that it preserves the concentration behaviour encoded in the less rare samples falling in sets of the form $\{\xx:\ell(\xx,\ttheta)\geq l\},$ for $l \ll v_\beta(\ttheta),$ so that the transformed samples $\mv{T}_h(\XX)$ replicate this concentration in the  target rare set  $\{\xx:\ell(\xx,\ttheta)\geq u\}$.  The respective sub-figures in the right panel in Figure \ref{fig:SS-Transform} show the distribution induced when the map $\mv{T}_h$ is applied to the less rare blue points in the left panel. In both the top and bottom panels in the right, observe that excess loss samples of $\ZZ = \mv{T}_h(\XX)$  falling in the target region (drawn in blue) concentrate in the same way as the target excess loss samples (drawn in red) concentrate.  However, the blue samples are observed 100 times more frequently than the respective red samples, thereby providing a reduction in the overall sample requirement. 
 
 A key ingredient in allowing the map $\mv{T}_h$ to possess this concentration-preserving property is the exponent $\mv{\kappa}(\cdot)$ defined in \eqref{eqn:kappa_def}. Briefly, $\mv{\kappa}$ ensures that the components of $\XX$ are only magnified/stretched to the extent necessary, as informed by the samples in the less rare region $\{\xx: \ell(\xx,\theta) \geq l\}.$  For instance, if a particular component of $\XX_i$ falling in this region is small, then the component remains small even after applying the map $\mv{T}_h$ (and vice versa). Please refer to \shortciteNP[Section 5]{deo2021achieving} for a definition and properties of rate-function preserving transformations, a notion which lends precision to the idea discussed only at an intuitive level here with the help of Figure \ref{fig:SS-Transform}. 



\subsection{Reduction in sample requirements in estimating the CVaR objective using IS } 
\label{sec:red-IS-samples}
As in Sections \ref{sec:IS_Background} - \ref{sec:SS-IS}, we fix a decision choice $\ttheta \in \Theta$ and discuss the magnitude of sample reductions one may obtain with the presented IS approaches in the estimation of CVaR $C_\beta(\ttheta).$
Let $n_{{\textsc{saa}}}(\beta)$ and $n_{{\tt is}}( \beta)$ denote the number of samples required by SAA and IS, respectively, in order to guarantee that the  objective $C_\beta(\ttheta)$ is estimated with a desired relative precision: In particular, the obtained estimate should be such that the relative error does not exceed a pre-specified level $\varepsilon \in (0,+\infty),$ with $(1-\alpha)\times 100\%$ confidence.  
\begin{theorem}\label{thm:SCR-Est}
For any $\delta > 0,$ we have the sample-requirement reduction guarantee that 
\begin{align*}
    \frac{n_{{\textsc{saa}}}(\beta)}{n_{\tt is}(\beta)} \geq \frac{c}{\beta^{1-\delta}}, \qquad \text{ for all } \beta < \beta_0,
\end{align*}
for  suitable constants $\beta_0 \in (0,1)$ and $c > 0,$ if either  
\begin{itemize}
    \item[(i)] IS density \eqref{eqn:exp_twise} based on exponential tilting is used and  Assumption  (1) below is satisfied by the loss $\ell(\cdot)$ and the distribution of $\XX;$ (or) 
    \item[(i)] IS estimator \eqref{eqn:FIS} based on self-structuring transformations is used and  Assumption  (2) below is satisfied by the loss $\ell(\cdot)$ and the distribution of $\XX.$
\end{itemize}
\end{theorem}
Theorem \ref{thm:SCR-Est} implies that the sample requirement due to IS grows only like $o((1/\beta^{\delta})),$ for any $\delta > 0,$ as the tail-level $\beta \rightarrow 0.$ This is in contrast to the steep $\tilde{O}(1/\beta)$ sample requirement that is inevitable with the use of SAA. Theorem~\ref{thm:SCR-Est} is a consequence of  the variance reduction guarantees presented in (i) \shortciteNP[Theorem 1]{arief2021certifiable} for exponential tilting and (ii) \shortciteNP[Theorem 1]{deo2021efficient}  for self-structuring IS transformations. 
Assumptions~\ref{assume:Exp_Twist_Requirement} and \ref{assume:BBIS} below, respectively, give the  conditions under the above sample reduction requirements can be guaranteed for IS using  the two approaches.

\begin{assumption}{(\textbf{Assumptions for Exponential Twisting).} }\label{assume:Exp_Twist_Requirement}
    \begin{enumerate}
        \item[(i)] The effective domain of the log-moment generating function of $\XX,$ $D(\Lambda)= \{\rr: \Lambda(\rr)<\infty\},$ has a non-empty interior. Moreover, $\Lambda(\cdot)$ is strictly convex and continuously differentiable in the interior of $D(\Lambda)$. 
        \item[(ii)] The set $\{\xx:\ell(\xx,\ttheta) \geq u\}$ is orthogonally monotone for any $\ttheta \in \Theta$: that is, if $\xx\leq\xx_1$ and $\ell(\xx,\ttheta) \geq u$ for $\ttheta \in \Theta,$ then it is necessary that $\ell(\xx_1,\ttheta)\geq u.$
    \end{enumerate}
\end{assumption}
Note that Part (i) of Assumption \ref{assume:Exp_Twist_Requirement} requires that $\XX$ is light-tailed. 
For stating the assumptions for self-structuring transformations, we introduce the following non-restrictive regularity notion: We call  a function $f(\cdot)$ to be \textit{multivariate regularly varying} if $f(n\xx)/f(n\mv{1})$ is uniformly convergent (as $n\to\infty$) in compact subsets of $\R^d.$
\begin{assumption}{(\textbf{Assumptions for Self-Structuring Transformations}).}\label{assume:BBIS}
\begin{enumerate}
    \item[(i)] Either the density $f_{\XX}(\cdot)$  or $\log f_{\XX}(\cdot)$ is multivariate regularly varying. 
    \item[ii)]  The loss $\ell(\cdot)$ satisfies the following for some $\rho > 0:$  $\ell(n\xx,\ttheta)/n^\rho$ is convergent uniformly in compact sets to a non-zero function. 
\end{enumerate}
\end{assumption}
Assumption \ref{assume:BBIS} allows both light and heavy-tailed distributions for $\XX.$ A wide variety of parametric and semiparametric multivariate distributions, including normal, exponential family, elliptical, log-concave distributions and Archimedian copula models  satisfy Assumption~\ref{assume:BBIS}(i) (see, \shortciteNP[Table 2]{deo2021achieving} for a comprehensive yet nonexhaustive collection of distributions satisfying Assumption \ref{assume:BBIS} or its more general variants). In its presented form, Assumption~\ref{assume:BBIS}(i) only allows for distributions whose marginals have a similar tail strength,
To see how the self-structuring transformation continues to work well even if this condition is relaxed, we refer the reader to \citeNP{deo2021achieving}. 
Part (ii) of Assumption~\ref{assume:BBIS} imposes a mild, non-parametric restriction on the how the loss function $\ell(\cdot)$ grows. It is satisfied easily in a number of losses which are of interest in stochastic optimization, including piecewise-linear losses, quadratic losses, and losses which occur in two-stage programs, such as in Example \ref{eg:two-stage-programs}. While the asymptotic sample requirement reduction in Theorem \ref{thm:SCR-Est} holds for any fixed choice of $h > 0$ in the transformation $\mv{T}_h,$ an ideal choice of the hyperparameter $h$ is made numerically as explained in Section \ref{sec:Retro_SAA} below. 

\section{RETROSPECTIVE APPROXIMATION as a solution paradigm}\label{sec:Retro_SAA} 
We examined in the previous section examples for how one may arrive at effective IS distributions for  any fixed decision choice $\ttheta \in \Theta.$  How to integrate the several IS distribution prescriptions we have, one for each decision choice, into a method for solving the optimization formulation \eqref{eqn:CVaR_OpT}? The solution paradigm of retrospective approximation (RA) serves as a natural vehicle to utilize the change of measure prescriptions effectively towards solving the CVaR minimization in  \eqref{eqn:CVaR_OpT}.

Retrospective approximation (\shortciteNP{CHEN_2001,pasupathy2010choosing}) has been developed as a computationally attractive alternative to SAA in solving general stochastic optimization and root finding problems. For a fixed error tolerance, recall that SAA involves solving  one large  problem formulated with all the samples allowed by the computational budget. The premise behind RA is to reduce the overall computational effort relative to SAA by instead solving a \textit{sequence} of SAA sub-problems: each sub-problem in the sequence is initialized with the solution of the previous sub-problem and is solved with a  larger number of samples, and upto a smaller error tolerance, than its predecessors.  This eases the overall computational burden as follows: The initial sub-problems are computationally light due to smaller sample sizes and larger error tolerances; the later sub-problems are computationally efficient as they are initialized with the solution from the previous stage, and refining them locally with the availability of more samples is less demanding than conducting an overall search. 

\subsection{Retrospective Approximation for Minimizing Tail Risks}
Interestingly, this general RA paradigm becomes an ideal vehicle for executing importance sampling, as every new sub-problem offers an opportunity to obtain samples from an IS distribution which is most suited for the regions in which the search for a solution is being presently conducted. 

With this guiding philosophy, the RA procedure outlined in Algorithm~\ref{algo:CVaR-I.S-RA-Oracle} below provides a  template to suitably incorporate IS distribution prescriptions we have from Section \ref{sec:IS} within a solution paradigm. In the general template presented in Algorithm~\ref{algo:CVaR-I.S-RA-Oracle}, this scheme assumes that we have chosen a suitable family  $\{g_{\alpha} : \alpha\in \mathcal{A}\}$ of IS probability densities, and a method for arriving at a good IS density choice $g_\alpha$ for any given selection of $(u,\ttheta)$ in solving the right hand side of \eqref{eqn:CVaR_OpT}. To capture this method's use succinctly in the algorithm, we represent it as an oracle mapping $\mv b: \R\times \Theta \rightarrow \mathcal{A}$ satisfying the following assumption.
\begin{assumption}
\label{assume:Oracle_Access}
The IS oracle  $\mv b: \R\times \Theta \rightarrow \mathcal{A}$ is such that for every decision $(u,\ttheta)$, the resulting IS probability density $g_\alpha,$  with  $\alpha$ set to $\alpha = \mv{b}(u,\ttheta)$, is the IS distribution of our choice for the estimation of $E(\ell(\XX,\ttheta) - u)^+$. 
\end{assumption}

\begin{algorithm}[h]
 \caption{Retrospective Approximation based CVaR Optimization}\label{algo:CVaR-I.S-RA-Oracle}
  \textbf{Input:}  Initial iterate $(u_0,\ttheta_0),$ an increasing sequence  $(n_k: k \geq 0)$ of sample-sizes with $n_0=0$, a decreasing sequence of error tolerances $\{\varepsilon_k : k \geq 1\},$ an initial IS parameter $\alpha_0$, IS oracle $\mv{b}(\cdot).$
  
 \noindent \textbf{ For $k\geq 1$, do}\\ 
  \textbf{1. Obtain importance samples:} Draw i.i.d. samples $\ZZ_{n_{k-1}+1},\ldots,\ZZ_{n_k}$ from the distribution $P_{\alpha_{k-1}}$.\\
\noindent \textbf{2. Solve the IS based optimization:} With the likelihood ratios set to  $\mathcal{L}_i = f_{\XX}(\ZZ_i)/g_{\alpha_{k-1}}(\ZZ_i)$ for $i=n_{k-1}+1,\ldots,n_k,$ solve the following problem upto a tolerance of $\varepsilon_k$:  \begin{equation}\label{eqn:CVaR-comp-Is-Oracle}
       \hat{c}_{is,n_k} := \inf_{u,\ttheta}\left[u + \frac{1}{n_k\beta}
      \sum_{i=1}^{n_k} \big(\ell({\mv{Z}}_{i},\ttheta) - u \big)^+
      \mathcal{L}_{i},\right] = \inf_{u,\ttheta} \hat{f}_{is,n_k}(u,\ttheta)
    \end{equation}
      with an initial solution $(u_{k-1},\ttheta_{k-1})$. Return $(u_k,\ttheta_k)$ as the solution obtained by solving \eqref{eqn:CVaR-comp-Is-Oracle}.\\
  \noindent \textbf{3. Update IS distribution choice:} Set $\alpha_{k} = \mv{b}(u_{k},\ttheta_{k})$. Set $k=k+1$.
   
  \noindent \textbf{Final Output: } Return the solution  $(u_k,\ttheta_k)$ 
  \end{algorithm}  

Describing a solution paradigm conveniently via an IS oracle, as in the case of Algorithm \ref{algo:CVaR-I.S-RA-Oracle} above, follows from the work of \shortciteNP{he2023adaptive}. Considering the case of self-structuring importance samplers explained in Section \ref{sec:SS-IS}, \shortciteNP{deo2022combining} provides the above RA procedure for integrating IS with optimization. The above RA procedure can alternatively be  interpreted as performing lazy-updates for IS distributions in the adaptive SAA scheme introduced by \shortciteNP{he2023adaptive}. 
\begin{example}[An IS oracle $\mv{b}$ for use with exponential tilting]\label{eg:Oracle_Twist}
\textnormal{
In the case of IS based on dominating points and exponential tilting considered in Example \ref{eg:Dominating_Points}, the following serve as the oracle mapping $\mv{b}(u,\ttheta):$ For any given $(u,\ttheta),$ consider the tilt parameters $(\mv{b}_i: i = 1,\ldots,M)$ obtained by solving \eqref{eqn:FL_Log_MGF_arg-mins} and \eqref{eqn:exp_twise}, together with the mixture weights $(p_i: i = 1,\ldots,M).$ With these  collections defining the parameters used in IS density in \eqref{eqn:exp_twise}, we use $(\mv{b}_i: i = 1,\ldots,M)$ and $(p_i: i = 1,\ldots,M)$ as the oracle mapping $\mv{b}(u,\ttheta)$
 }
\end{example}
\begin{example}[An IS oracle $b(\cdot)$ for use with self-structuring IS in Section \ref{sec:SS-IS}]
\label{eg:oracle-SS-IS}
\em
Recall that for the case of self-structuring IS transformations introduced in Section \ref{sec:SS-IS}, a wide range of stretching hyperparameters $h$ have been shown to offer good variance reduction asymptotically. If one wishes to update to a specific choice of $h$ in Step 3 of Algorithm \ref{algo:CVaR-I.S-RA-Oracle} above, they may do so via a simple cross-validation type one-dimensional search demonstrated in Algorithm \ref{algo:CVaR-FULL.} below. 
\begin{algorithm}[h]
 \caption{An IS oracle $\mv{b}$ for use with Self-Structuring IS transformations}
 \label{algo:CVaR-FULL.}
  \textbf{Input:}  iterate $(u,\ttheta)\,$, i.i.d. samples $\XX_{1},\ldots,\XX_m$ from the distribution of $\XX$,  initial seed $h_0$.\\

 \noindent \textbf{1. The oracle objective}, for any choice of stretch hyperparameter $h,$ is evaluated to be the second moment defined below: 
\begin{equation}\label{eqn:prob_h}
     \hat{M}_{2}(h\,;\,u,\ttheta) =  \frac{1}{m} \sum_{i=1}^m \left[\left( \ell(\XX_{i};\ttheta ) - u \right)^+ \right]^2 \mathcal{L}_{h,i}, 
 \end{equation}
 where $\mathcal{L}_{h,i}$ denotes the likelihood ratio \eqref{eq:LR}.
   
\noindent \textbf{2 Update the cross validation parameter: }  Return a  stretch parameter minimizing the oracle objective above as the output of the IS oracle:
\begin{equation}\label{eqn:h-grad}
       \mv{b}(u,\ttheta) \  \in \  \arg\min_{h}  \, \hat{M}_{2}(h\,;\,u,\ttheta) 
    \end{equation}
  \end{algorithm}  
\end{example}

\subsection{Guidance on selecting sample size and tolerance parameter for RA}
Assumption~\ref{assume:RA-size} below imposes a mild condition  on the sample-size $n_k$ and error tolerance $\varepsilon_k$ one may need to use for the $k$-th sub-problem in Algorithm \ref{algo:CVaR-I.S-RA-Oracle}. These conditions follow from the guidance derived in \citeNP[Assumptions C.1-C.3]{pasupathy2010choosing}.

\begin{assumption}
\label{assume:RA-size}\em
Suppose that the sequence $\{(\varepsilon_k,n_k): k \geq 1\}$ satisfies the following requirements: 
\begin{enumerate}
    \item If the optimization procedure used to solve \eqref{eqn:CVaR-comp-Is-Oracle} exhibits linear convergence, then  $(n_k,\varepsilon_k : k \geq 1)$ is such that  $\liminf_{k\to\infty}\varepsilon_{k-1} \sqrt{n_k} >0$. If this procedure exhibits polynomial convergence, then  $(n_k,\varepsilon_k : k \geq 1)$ is such that  $\liminf_{k\to\infty} {\log 1/\sqrt{n_{k-1}}}({\log \varepsilon_k})^{-1} >0$.
    \item $\limsup_{k\to\infty} (\sum_{j=1}^k n_j)^2/\varepsilon_{k}^2<\infty $ and $\limsup_{k\to\infty} n_k^{-1}\sum_{j=1}^k n_j <\infty$.
\end{enumerate}
\end{assumption}
Assumption~\ref{assume:RA-size} imposes conditions so that the errors due to finite sample size and the errors due to solver error tolerance are balanced out, so that the cumulative work performed is kept minimal. 
For instance,  condition 1 requires that the optimization error tolerance $\varepsilon_k$ decays does not decay to $0$ too fast. If this is not satisfied, then the error in solution due to imperfect optimization will be orders of magnitude smaller than the sampling error, and therefore lead to a wastage of computational effort.
Likewise, achieving  low variance with a large sample size while allowing a larger error tolerance in the solver is also computationally inefficient. 
Conditions 2 imposes a lower bound on rates at which $(n_k,\varepsilon_k)$ converge to their limits, so that the solutions output by successive epochs do not get ``stuck". 
For instance, if $\varepsilon_k$ converges to $0$ too slowly, then the tolerance condition may be too easily satisfied, and therefore lead to no improvement in solution. 
Conversely if $n_k$ goes to $\infty$ too slowly, the difference in $n_{k-1}$ and $n_k$ is so small that the iterate does not move.  In either case, the work done in the $k$th iteration of the RA procedure is wasted and leads to a  computational suboptimality. Thus the specifications of sample size and error tolerance are such that the two errors are balanced just about ideally in every sub-problem \eqref{eqn:CVaR-comp-Is-Oracle}.

A natural choice sample size is $n_{k} = \lceil cm_{k-1}\rceil$ for linearly converging  optimization procedures, and  $n_k = \lceil n_{k-1}^c\rceil$ for polynominally converging procedures. Meanwhile, for both these cases, $\varepsilon_k  = K/\sqrt{n_k}$ is a good choice for the error tolerance. We refer interested readers to \shortciteNP[Section 5.3]{pasupathy2010choosing} for a detailed investigation.

\subsection{Reduction in sample requirements due to IS}
Recall the IS oracle $\mv{b}(\cdot)$ mapping selections  illustrated in Examples \ref{eg:Oracle_Twist} - \ref{eg:oracle-SS-IS}. For these IS oracles, we next show that the reduction in sample requirement exhibited in Theorem \ref{thm:SCR-Est} for a fixed decision  carries forward to solving the risk minimisation problem \eqref{eqn:CVaR_OpT}. As before, let $n_{\textsc{saa}}(\beta)$ and $n_{{\tt is}}(\beta)$ denote the number of samples required to optimise CVaR such that the resulting optimal values lie within $\varepsilon$-relative precision of the true optimal value $c_\beta$ in \eqref{eqn:CVaR_OpT} with $(1-\alpha) \times 100\%$ confidence. 
\begin{theorem}\label{prop:IS-RA-sample-error}
Suppose that Assumptions \ref{assume:Oracle_Access} - \ref{assume:RA-size} are satisfied. Then for any $\delta > 0,$ we have 
\begin{align*}
    \frac{n_{{\textsc{saa}}}(\beta)}{n_{\tt is}(\beta)} \geq \frac{c}{\beta^{1-\delta}}, \qquad \text{ for all } \beta < \beta_0,
\end{align*}
for  suitable constants $\beta_0 \in (0,1)$ and $c > 0,$ if either  (i) Assumption~\ref{assume:Exp_Twist_Requirement} holds and the oracle mapping for Algorithm~\ref{algo:CVaR-I.S-RA-Oracle} be as in Example~\ref{eg:Oracle_Twist}; or (ii) Assumption~\ref{assume:BBIS} holds and the parameter $h$ for the self-structuring IS transformation $\mv{T}_h(\cdot)$ be selected as in Algorithm~\ref{algo:CVaR-FULL.}. 
\end{theorem}

Similar to Theorem \ref{thm:SCR-Est}, Theorem~\ref{prop:IS-RA-sample-error} implies that the sample requirement to solve the more challenging CVaR optimization problem using IS also grows only like $o(1/\beta^{\delta}),$ for any $\delta > 0,$ as the tail-level $\beta \rightarrow 0.$ In particular, there is no efficiency lost in embedding the IS change of distributions using the RA solution paradigm. To get a numerical sense of the savings in sample requirement, consider the portfolio optimization task in Example \ref{eg:lin-portfolio}. For  $\beta=1/100$, \shortciteNP{deo2022combining} demonstrates, for instance, that while SAA takes about $8000$ samples to achieve a $1\%$ relative error, self-structuring IS require only about $550$ samples (about 15 times less).  \shortciteNP{deo2022combining} carries an additional analysis on the savings in the  work complexity (or the total computational effort) due to the RA 
 procedure in Algorithm \ref{algo:CVaR-FULL.} relative to SAA.

\section{STOCHASTIC APPROXIMATION as a solution paradigm}\label{sec:Adaptive_IS}
As an alternative to RA, one may also consider iterative stochastic approximation methods as a solution paradigm for solving \eqref{eqn:CVaR_OpT}.
Under mild technical conditions, \eqref{eqn:CVaR_OpT} reduces to the stochastic root finding problem
\begin{equation}\label{eqn:SPO_FOCs}
    \nabla f(u,\ttheta) = \mv{0} \text{ or equivalently } E[\mv{G}(\XX;u,\ttheta)] = \mv{0} \text{ where } \mv{G}(\xx;u,\ttheta) = \frac{\partial F (\xx;u,\ttheta).}{\partial (u,\ttheta)} 
\end{equation}
A typical approach to solving \eqref{eqn:SPO_FOCs}, without any change of distribution, is to use the following iterative scheme (see, for e.g.., \shortciteNP{AGbook}): Given a sample $\XX_n,$  update the iterate recursively via 
\begin{equation}\label{eqn:Grad_Step}
    (u_n,\ttheta_n) \leftarrow (u_{n-1},\ttheta_{n-1}) - \gamma_n \mv{G}(\XX_n;u_{n-1},\ttheta_{n-1}) \text{ for } n=1,2,\ldots 
\end{equation}
The gradients  $\mv{G}(\xx;u,\ttheta) = (\frac{\partial F}{\partial u}, \frac{\partial F}{\partial \ttheta})(\xx\,;\, u,\ttheta)$ can be readily computed for the objective \eqref{eqn:CVaR_OpT} as in,
\begin{equation}\label{eqn:Grad_u}
\frac{\partial }{\partial u} F(\xx\,;\, u,\ttheta)  =  1- \beta^{-1}\setindicator(\ttheta^\intercal\xx \geq u)\quad
  \frac{\partial }{\partial \ttheta} F(\xx\,;\, u,\ttheta) =  \beta^{-1} \frac{\partial \ell}{\partial \ttheta} (\xx; \ttheta) \,\setindicator(\ttheta^\intercal\xx \geq u)
   \end{equation}
Following the approach devised in \shortciteNP{he2023adaptive}, one may readily embed the IS change of distributions into the update step as shown in  Algorithm~\ref{algo:Oracle_Based_IS_CVaR} below.  One may also obtain the sample-requirement reduction guarantees similar to that obtained for the retrospective approximation in Theorem \ref{prop:IS-RA-sample-error}, thanks to the convergence analysis executed in \shortciteNP{he2023adaptive}. 
\begin{algorithm}[h] 
  \caption{CVaR Optimization using oracle-based adaptive IS}\label{algo:Oracle_Based_IS_CVaR}
  \ \vspace{-2pt}\\
  \KwIn{Density $f_{\XX}(\cdot),$  initial iterate $(u_0,\ttheta_0)$, initial IS parameter choice  $\alpha_0$, IS oracle  $\mv{b}$, step-size parameters $c > 0,\ \gamma \in (1/2,1)$ }
     Initialise $n=1$. \textbf{While} stopping criterion not met \textbf{do} 
     \begin{enumerate}
         \item Generate an independent sample $\ZZ_n$ from the IS density $g_{\alpha_n}(\cdot)$.
         \item Set $\gamma_n = cn^{-\gamma}$ and update root estimate 
         \begin{equation}\label{eqn:Robbins_IS}
             (u_n,\ttheta_n) \leftarrow (u_{n-1},\ttheta_{n-1}) - \gamma_n \mv{G} (\ZZ_{n}; u_{n-1},\ttheta_{n-1}) \frac{f_{\XX}(\ZZ_n)}{g_{\alpha_n}(\ZZ_n)}
         \end{equation}

        
         \item Update IS parameter $\alpha_{n+1} = \mv{b}(u_{n},\ttheta_{n})$. Set $n=n+1$.
     \end{enumerate}
     \textbf{Output:} Return the averaged iterate   $(\bar{u}_n,\bar{\ttheta}_n) = n^{-1}\sum_{i=1}^n(u_{n},\ttheta_{n})$ as an estimate of the optimal solution to \eqref{eqn:CVaR_OpT}.
    \end{algorithm}

\section*{Acknowledgements}
The authors acknowledge support from Singapore Ministry of Education grant MOE2019-T2-2-163.

{ \small \bibliographystyle{wsc}
\bibliography{demobib}}

\end{document}